\documentclass{article}

\usepackage{arxiv}

\usepackage[utf8]{inputenc} 
\usepackage[T1]{fontenc}    
\usepackage{hyperref}       
\usepackage{url}            
\usepackage{booktabs}       
\usepackage{amsfonts}       
\usepackage{nicefrac}       
\usepackage{microtype}      
\usepackage{lipsum}		
\usepackage{graphicx}
\usepackage{natbib}
\usepackage{doi}
\usepackage{amsmath} 
\usepackage{amsthm}
\usepackage{amssymb}
\usepackage{amsfonts}
\usepackage{enumitem}
\usepackage{float}
\usepackage{tabularx}

\newtheorem{theorem}{Theorem}

\newtheorem{definition}{Definition}

\newcolumntype{C}{>{\centering\arraybackslash}X}

\title{Orders Between Channels and Implications for Partial Information Decomposition}

\date{} 					

\author{{\hspace{1mm}André F. C. Gomes}\thanks{Corresponding author. } \\
	Instituto de Telecomunicações\\
	Instituto Superior Técnico\\
	Universidade de Lisboa, Portugal \\
	\texttt{andrefcgomes@tecnico.ulisboa.pt} \\
	\And
	{\hspace{1mm}Mário A. T. Figueiredo} \\
	Instituto de Telecomunicações\\
	Instituto Superior Técnico\\
	Universidade de Lisboa, Portugal \\
	\texttt{mario.figueiredo@tecnico.ulisboa.pt} \\
}

\renewcommand{\shorttitle}{Gomes and Figueiredo}

\hypersetup{
pdftitle={Orders between channels and implications for partial information decomposition},
pdfsubject={},
pdfauthor={André F.C. Gomes, Mário A.T. Figueiredo},
pdfkeywords={information theory, partial information decomposition, channel preorders, intersection information, redundancy},
}

\begin{document}
\maketitle

\begin{abstract}
The \textit{partial information decomposition} (PID) framework is concerned with decomposing the information that a set of random variables has with respect to a target variable into three types of components: redundant, synergistic, and unique. Classical information theory alone does not provide a unique way to decompose information in this manner and additional assumptions have to be made. Recently, Kolchinsky proposed a new general axiomatic approach to obtain measures of redundant information, based on choosing a preorder relation between information sources (equivalently, a preorder between communication channels). In this paper, we exploit this approach to introduce three new measures of redundant information (and the resulting decompositions) based on well-known preorders between channels, thus contributing to the enrichment of the PID landscape. We relate the new decompositions to existing ones, study some of their properties, and provide examples illustrating their novelty. As a side result, we prove that any preorder that satisfies Kolchinsky's axioms yields a decomposition that meets the axioms originally introduced by Williams and Beer when they first propose the PID.
\end{abstract}


\section{Introduction}
\citet{williams2010nonnegative} proposed the \textit{partial information decomposition} (PID) framework as a way to characterize, or analyze, the information that a set of random variables (often called \textit{sources}) has about another variable (referred to as the \textit{target}). PID is a useful tool for gathering insights and analyzing the way information is stored, modified, and transmitted within complex systems \citep{lizier2013towards}, \citep{wibral2017quantifying}. It has found applications in areas such as cryptography \citep{rauh2017secret} and neuroscience \citep{vicente2011transfer,ince2015tracing}, with many other potential use cases, such as in understanding how information flows in gene regulatory networks \citep{gates2016control}, neural coding \citep{faber2019computation}, financial markets \citep{james2018modes}, and network design \citep{arellano2013shannon}.

Consider the simplest case: a three-variable joint distribution $p(y_1, y_2,t)$ describing three random variables: two sources, $Y_1$ and $Y_2$, and a target $T$. Notice that, despite what the names \textit{sources} and \textit{target} might suggest, there is no directionality (causal or otherwise) assumption. The goal of PID is to \textit{decompose} the information that $Y=(Y_1, Y_2)$ has about $T$ into the sum of 4 non-negative quantities: the information that is present in both $Y_1$ and $Y_2$, known as \textit{redundant} information $R$; the information that only $Y_1$ (respectively $Y_2$) has about $T$, known as \textit{unique} information $U_1$ (respectively $U_2$); the \textit{synergistic} information $S$ that is present in the pair $(Y_1, Y_2)$ but not in $Y_1$ or $Y_2$ alone. That is, in this case with two variables, the goal is to write
\begin{equation} \label{decomposition}
I(T;Y) = R + U_1 + U_2 + S,
\end{equation}
where $I(T;Y)$ is the mutual information between $T$ and $Y$ \citep{cover1999elements}. Because unique information and redundancy satisfy the relationship  $U_i = I(T; Y_i) - R$ (for $i 
\in\{1,2\}$), it turns out that defining how to compute one of these quantities ($R$, $U_i$, or $S$) is enough to fully determine the others \citep{williams2010nonnegative}. As the number of variables grows, the number of terms appearing in the PID of $I(T;Y)$ grows super exponentially \citep{gutknecht2021bits}. \citet{williams2010nonnegative} suggested a set of axioms that a measure of redundancy should satisfy, and proposed a measure of their own. Those axioms became known as the Williams-Beer axioms and the measure they proposed has subsequently been criticized for not capturing informational content, but only information size \citep{harder2013bivariate}.


Spawned by that initial work, other measures and axioms for information decomposition have been introduced; see, for example, the work by \citet{bertschinger2014quantifying}, \citet{griffith2014quantifying}, and \citet{james2018unique}. There is no consensus about what axioms any measure should satisfy or whether a given measure is \textit{capturing the information} that it should capture, except for the Williams-Beer axioms. Today, there is still debate about what axioms a measure of redundant information should satisfy and there is no general agreement on what is an appropriate PID \citep{chicharro2017synergy, james2018unique, bertschinger2013shared, rauh2017coarse, ince2017measuring}. 

Recently, \citet{kolchinsky2022novel} suggested a new general approach to define measures of redundant information, also known as \textit{intersection information} (II), the designation that we adopt hereinafter. At the core of that approach is the choice of a preorder relation between information sources (random variables), which allows comparing two sources in terms of how informative they are with respect to the target variable.  

In this work, we take previously studied preorders between communication channels, which correspond to preorders between the corresponding output variables in terms of information content with respect to the input. Following Kolchinsky's approach, we show that those preorders thus lead to the definition of new II measures. The rest of the paper is organized as follows. In Sections 2 and 3, we review Kolchinsky's definition of an II measure and the \textit{degradation} preorder. In Section 4, we describe some preorders between channels, based on the work by \citet{krner1977comparison} and \citet{americo2021channel}, derive the resulting II measures, and study some of their properties. Section 5 presents and comments on the optimization problems involved in the computation of the proposed measures. In Section 6, we explore the relationships between the new II measures and previous PID approaches, and we apply the proposed II measures to some famous PID problems. Section 7 concludes the paper by pointing out some suggestions for future work.

\section{Kolchinsky's Axioms and Intersection Information}
Consider a set of $n$ discrete random variables, $Y_1 \in \mathcal{Y}_1, ..., Y_n\in\mathcal{Y}_n$, called the source variables, and let $T \in \mathcal{T}$ be the (also discrete) target variable, with joint distribution (probability mass function) $p(y_1, ..., y_n, t)$.  Let $\preceq$ denote some preorder between random variables that satisfies the following axioms, herein referred to as \emph{Kolchinsky's axioms} \citep{kolchinsky2022novel}:
\begin{description}
        \item[(i)] Monotonicity of \textit{mutual information}\footnote{In this paper, mutual information is always assumed as referring to Shannon's mutual information, which, for two discrete variables $X \in \mathcal{X}$ and $Z\in \mathcal{Z}$, is given by 
    \[
    I(X;Z) = \sum_{x\in \mathcal{X}} \sum_{z\in \mathcal{Z}} p(x,z) \log\frac{p(x,z)}{p(x)\, p(z)}, 
    \]
    and satisfies the following well-known fundamental properties: $I(X;Z) \geq 0$ and $I(X;Z)=0 \,\Leftrightarrow \, X \perp Z$ ($X$ and $Z$ are independent) \citep{cover1999elements}.} w.r.t. $T$: $Y_i \preceq Y_j \, \Rightarrow \, I(Y_i;T) \leq I(Y_j;T)$.
    \item[(ii)] Reflexivity: $Y_i \preceq Y_i$ for all $Y_i$.
    \item[(iii)] For any $Y_i$, $C \preceq Y_i \preceq (Y_1, ..., Y_n)$, where $C\in \mathcal{C}$ is 
    any variable taking a constant value with probability one, \textit{i.e.}, with a distribution that is a delta function or such that $\mathcal{C}$ is a singleton.
\end{description}
\citet{kolchinsky2022novel} showed that such a preorder can be used to define an II measure via
\begin{equation}\label{II}
I_\cap (Y_1, ..., Y_n \rightarrow T) := \sup_{Q:\; Q \preceq Y_i, \; i\in\{1,..,n\} } I(Q;T),
\end{equation}
and we now show that if $\preceq$ is a preorder then the II measure in \eqref{II} satisfies the Williams-Beer axioms \citep{williams2010nonnegative, lizier2013towards}, thus establishing a strong connection between these formulations. Recall that a relation $\preceq$ is a preorder if it satisfies transitivity and reflexivity \citep{schroder2003ordered}. Before stating and proving this result, we recall the Williams-Beer axioms \citep{lizier2013towards}.

\begin{definition}
Let $A_1, ..., A_r$ an arbitrary number of $r \geq 2$ sources \footnote{In the context of \citep{lizier2013towards}, the definition of a source $A_i$ is that of a set of random variables, e.g., $A_1 = \{Y_1, Y_2\}$.}. An intersection information measure $I_\cap$ is said to satisfy the Williams-Beer axioms if it satisfies:
\begin{enumerate}
    \item Symmetry: $I_\cap$ is symmetric in the $A_i$'s.
    \item Self-redundancy: $I_\cap(A_i) = I(A_i;T)$.
    \item Monotonicity: $I_\cap(A_1, ..., A_{r-1}, A_r) \leq I_\cap(A_1, ..., A_{r-1})$.
    \item Equality for Monotonicity: If $A_{r-1} \subseteq A_r$ then $I_\cap(A_1, ..., A_{r-1}, A_r) = I_\cap(A_1, ..., A_{r-1})$
\end{enumerate}
\end{definition}

\begin{theorem}
Let $\preceq$ be some preorder that satisfies Kolchinsky's axioms and define its corresponding II measure as in \eqref{II}. Then the corresponding II measure satisfies the Williams-Beer axioms.
\end{theorem}

\begin{proof}
Symmetry and monotonicity follow trivially given the form of \eqref{II} (the definition of supremum and restriction set). Self-redundancy follows from the reflexivity of the preorder and monotonicity of mutual information. Now suppose $A_{r-1} \subseteq A_r$, and let $Q$ be any admissible point of $I_\cap(A_1, ..., A_{r-1})$, which implies that $Q \preceq A_{r-1}$. Now, since $A_{r-1} \subseteq A_r$, the Kolchinsky axiom (iii) and transitivity of the preorder $\preceq$ guarantees that $Q \preceq\ A_{r-1} \preceq A_r$, which means that $Q$ is an admissible point of $I_\cap(A_1, ..., A_{r})$. Therefore $I_\cap(A_1, ..., A_{r-1}, A_r) \geq I_\cap(A_1, ..., A_{r-1})$ and monotonicity guarantees that $I_\cap(A_1, ..., A_{r-1}, A_r) = I_\cap(A_1, ..., A_{r-1})$.
\end{proof}

In conclusion, every preorder that satisfies the set of axioms introduced by \citet{kolchinsky2022novel} yields a valid II measure, in the sense that the measure satisfies the Williams-Beer axioms. Having a \emph{more informative} relation $\preceq$ allows us to draw conclusions about information flowing from different sources. It also allows for the construction of PID measures that are well-defined for more than two sources. In the following, we will omit ``$\rightarrow T$" from the notation (unless we need to explicitly refer to it), with the understanding that the target variable is always some arbitrary, discrete random variable $T$.

\section{Channels and the Degradation/Blackwell Preorder}
Given two discrete random variables $X \in \mathcal{X}$ and $Z\in \mathcal{Z}$, the corresponding conditional distribution $p(z|x)$ corresponds, in an information-theoretical perspective, to a discrete memoryless channel with a channel matrix $K$, \textit{i.e.}, such that $K[x,z] = p(z|x)$ \citep{cover1999elements}. This matrix is row-stochastic: $K[x,z] \geq 0$, for any  $x \in \mathcal{X}$ and $z\in \mathcal{Z}$, and $\sum_{z \in \mathcal{Z}} K[x,z]=1$

The comparison of different channels (equivalently, different stochastic matrices) is an object of study with many applications in different fields \citep{cohen1998comparisons}. That study addresses preorder relations between channels and their properties. One such preorder, named \textit{degradation preorder} (or \textit{Blackwell preorder}) and defined next, was used by Kolchinsky to obtain a particular II measure \citep{kolchinsky2022novel}.

Consider the distribution $p(y_1, ..., y_n, t)$ and the channels $K^{(i)}$ between $T$ and each $Y_i$, that is, $K^{(i)}$ is a $|\mathcal{T}| \times |\mathcal{Y}_i|$ row-stochastic matrix with the conditional distribution $p(y_i|t)$.

\begin{definition}
We say that channel $K^{(i)}$ is a \emph{degradation} of channel $K^{(j)}$, and write $K^{(i)}\preceq_{d} K^{(j)}$ or $Y_i \preceq_{d} Y_j$, if there exists a channel $K^U$ from $Y_j$ to $Y_i$, \textit{i.e.}, a 
$|\mathcal{Y}_j| \times |\mathcal{Y}_i|$ row-stochastic matrix, such that $K^{(i)} = K^{(j)}K^U$.
\end{definition}
Intuitively, consider 2 agents, one with access to $Y_i$ and the other with access to $Y_j$. The agent with access to $Y_j$ has at least as much information about $T$ as the one with access to $Y_i$, because it has access to channel $K^U$, which allows sampling from $Y_i$, conditionally on $Y_j$ \citep{rauh2017coarse}. \citet{blackwell1953equivalent} showed that this is equivalent to saying that, for whatever decision game where the goal is to predict $T$ and for whatever utility function, the agent with access to $Y_i$ cannot do better, on average, than the agent with access to $Y_j$.

Based on the degradation/Blackwell preorder, \citet{kolchinsky2022novel} introduced the \textit{degradation II measure}, by plugging the ``$\preceq_{d}$" preorder in \eqref{II}:
\begin{equation}\label{d}
I_\cap^d (Y_1, ..., Y_n) := \sup_{Q:\; Q \preceq_d Y_i, \; i\in\{1,..,n\} } I(Q;T).
\end{equation}
As noted by \citet{kolchinsky2022novel}, this II measure has the following operational interpretation. Suppose $n=2$ and consider agents 1 and 2, with access to variables $Y_1$ and $Y_2$, respectively. Then $I_\cap^d (Y_1, Y_2)$ is the maximum information that agent 1 (resp. 2) can have w.r.t. $T$ without being able to do better than agent 2 (resp. 1) on any decision problem that involves guessing $T$. That is, the degradation II measure quantifies the existence of a dominating strategy for any guessing game.

\section{Other Preorders and Corresponding II Measures}
\subsection{The ``Less Noisy" Preorder}
\citet{krner1977comparison} introduced and studied preorders between channels with the same input. We follow most of their definitions and change others when appropriate. We interchangeably write $Y_1 \preceq Y_2$ to mean $K^{(1)} \preceq K^{(2)}$, where $K^{(1)}$ and $ K^{(2)}$ are the channel matrices as defined above.

Before introducing the next channel preorder, we need to review the notion of Markov chain \citep{cover1999elements}. We say that three random variables $X_1$, $X_2$, $X_3$ form a Markov chain, and write $X_1 \rightarrow X_2 \rightarrow X_3$, if the following equality holds: $p(x_1,x_3|x_2) = p(x_1|x_2)\, p(x_3|x_2)$, \textit{i.e.}, if $X_1$ and $X_3$ are conditionally independent, given $X_2$. Of course, $X_1 \rightarrow X_2 \rightarrow X_3$ if and only if $X_3 \rightarrow X_2 \rightarrow X_1$.

\begin{definition}
We say that channel $K^{(2)}$ is \emph{less noisy} than channel $K^{(1)}$, and write $K^{(1)}\preceq_{ln} K^{(2)}$, if for any discrete random variable $U$ with finite support, such that both $U \rightarrow T \rightarrow Y_1$ and $U \rightarrow T \rightarrow Y_2$ hold, we have that $I(U;Y_1) \leq I(U;Y_2)$.
\end{definition}

The \textit{less noisy} preorder has been primarily used in network information theory to study the capacity regions of broadcast channels \citep{makur2017less} and the secrecy capacity of the wiretap and eavesdrop channels problem \citep{{csiszar2011information}}. Secrecy capacity ($C_S$) is the maximum rate at which information can be transmitted over a communication channel while keeping the communication secure from eavesdroppers - that is, having zero information leakage \citep{wyner1975wire, bassi2019secret}. It has been shown that $C_S > 0$ unless $K^{(2)}\preceq_{ln} K^{(1)}$, where $C_S$ is the secrecy capacity of the Wyner wiretap channel with $K^{(2)}$ as the main channel and $K^{(1)}$ as the eavesdropper channel \citep[Corollary 17.11]{csiszar2011information}.

Plugging the \textit{less noisy} preorder $\preceq_{ln}$ in \eqref{II} yields a new II measure
\begin{equation}\label{ln}
I_\cap^{ln} (Y_1, ..., Y_n) := \sup_{Q:\; Q \preceq_{ln} Y_i, \; i\in\{1,..,n\} } I(Q;T).
\end{equation}
Intuitively, $I_\cap^{ln} (Y_1, ..., Y_n)$ is the most information that a channel $K^Q$ can have about $T$ such that any other channel $K^{(i)}, i=1, ..., n$, is \emph{less noisy} than it, that is, a channel that leads to a zero secrecy capacity when any other channel $K^{(i)}$ is the eavesdropper.


\subsection{The ``More Capable" preorder}
The next preorder we consider, termed ``\textit{more capable}", was used in calculating the capacity region of broadcast channels \citep{gamal1979capacity} or in deciding whether a system is more secure than another \citep{clark2005quantitative}. See the book by \citet{cohen1998comparisons}, for more applications of the \textit{degradation, less noisy}, and \textit{more capable} preorders.

\begin{definition}
We say that channel $K^{(2)}$ is \emph{more capable} than $K^{(1)}$, and write $K^{(1)}~\preceq_{mc}~K^{(2)}$, if for any distribution $p(t)$ we have $I(T;Y_1) \leq I(T;Y_2)$.
\end{definition}

Inserting the ``\textit{more capable}" preorder into \eqref{II} leads to 
\begin{equation}\label{mc}
I_\cap^{mc} (Y_1, ..., Y_n) := \sup_{Q:\; Q \preceq_{mc} Y_i, \; i\in\{1,..,n\} } I(Q;T),
\end{equation}
that is, $I_\cap^{mc} (Y_1, ..., Y_n)$ is the information that the `largest' (in the more capable sense), but no larger than any $Y_i$, random variable $Q$ has w.r.t. $T$. Whereas under the degradation preorder, it is guaranteed that if $Y_1 \preceq_d Y_2$, then agent 2 will make better decisions, for whatever decision game, on average, under the ``more capable" preorder such a guarantee is not available. We do, however, have the guarantee that if $Y_1 \preceq_{mc} Y_2$, then for whatever distribution $p(t)$ we know that agent 2 will always have more information about $T$ than agent 1. This has an interventional approach meaning: if we intervene on variable $T$ by changing its distribution $p(t)$ in whichever way we see fit, we have that $I(Y_1;T) \leq I(Y_2;T)$ (assuming that the distribution $p(Y_1, ..., Y_n, T)$ can be modeled as a set of channels from $T$ to each $Y_i$). That is, $I_\cap^{mc} (Y_1, ..., Y_n)$ is the highest information that a channel $K^Q$ can have about $T$ such that for any change in $p(t)$, $K^Q$ knows less about $T$ than any $Y_i, i = 1, ..., n$. Since PID is concerned with decomposing a distribution that has fixed $p(t)$, the ``more capable" measure is concerned with the mechanism by which $T$ generates $Y_1, ..., Y_n$, for any $p(t)$, and not concerned with the specific distribution $p(t)$ yielded by $p(Y_1, ..., Y_n, T)$.

For the sake of completeness, we could also study the II measure that would result from the capacity preorder. Recall that the capacity of the channel from a variable $X$ to another variable $Z$, which is only a function of the conditional distribution $p(z|x)$, is defined as \citep{cover1999elements}
\begin{equation} 
C = \max_{p(x)} I(X;Z). \label{eq:capacity}    
\end{equation}
\begin{definition}
Write $W \preceq_{c} V$ if the capacity of V is at least as large as the capacity of W.
\end{definition}
Even though it is clear that $W \preceq_{mc} V \; \Rightarrow  \; W \preceq_c V$, the $\preceq_{c}$ preorder does not comply with the first of Kolchinsky's axioms (since the definition of capacity involves the choice of a particular marginal that achieves the maximum in \eqref{eq:capacity}, which may not coincide with the marginal corresponding to $p(y_1,...,y_n,t)$), which is why we don't define an II measure based on it. 

\subsection{The ``Degradation/Supermodularity" preorder}
In order to introduce the last II measure, we follow the work and notation of  \citet{americo2021channel}. Given two real vectors $r$ and $s$ with dimension $n$, let $r \lor s := (\max(r_1, s_1), ..., \max(r_n, s_n))$ and $r \land s := (\min(r_1, s_1), ..., \min(r_n, s_n))$. Consider an arbitrary channel $K$ and let $K_i$ be its $i$th column. From $K$, we may define a new channel, which we construct column by column using the \emph{JoinMeet} operator $\diamond_{i,j}$. Column $l$ of the new channel is defined, for $i \neq j$, as
$$(\diamond_{i,j}K)_l = \begin{cases} K_i \lor K_j , & \mbox{if } l = i \\ K_i \land K_j, & \mbox{if } l = j \\ K_l, & \emph{otherwise} \end{cases}.$$
\citet{americo2021channel} used this operator to define the following two new preorders.
Intuitively, the operator $\diamond_{i,j}$ makes the rows of the channel matrix more similar to each other, by putting in column $i$ all the maxima and in column $j$ the minima, between every pair of elements in columns $i$ and $j$ of every row. In the following definitions, the \textit{s} stands for supermodularity, a concept we need not introduce in this work.

\begin{definition}
We write $W \preceq_{s} V$ if there exists a finite collection of tuples $(i_k, j_k)$ such that $W = \diamond_{i_1,j_1} ( \diamond_{i_2,j_2} ( ... (\diamond_{i_m,j_m} V ))$.
\end{definition}

\begin{definition}
Write $W \preceq_{ds} V$ if there are $m$ channels $U^{(1)}, ..., U^{(m)}$ such that $W \preceq_0 U^{(1)} \preceq_1 U^{(2)} \preceq_2 ... \preceq_{m-1} U^{(m)} \preceq_m V$, where each $\preceq_i$ stands for $\preceq_d$ or $\preceq_s$. We call this the \emph{degradation/supermodularity} preorder.
\end{definition}

Using the ``\emph{degradation/supermodularity}" (ds) preorder, we define the ds II measure as:
\begin{equation}\label{ds}
I_\cap^{ds} (Y_1, ..., Y_n) := \sup_{Q:\; Q \preceq_{ds} Y_i, \; i\in\{1,..,n\} } I(Q;T).
\end{equation}
The \emph{ds} preorder was recently introduced in the context of \textit{core-concave entropies} \citep{americo2021channel}. Given a core-concave entropy $H$, the \textit{leakage} about $T$ through $Y_1$ is defined as $I_H(T;Y_1) = H(T)-H(T|Y_1)$ . In this work, we are mainly concerned with Shannon's entropy $H$, but as we will elaborate in the future work section below, one may apply PID to other core-concave entropies. Although the operational interpretation of the \emph{ds} preorder is not yet clear, it has found applications in privacy/security contexts and in finding the most secure deterministic channel (under some constraints) \citep{americo2021channel}.

\subsection{Relations Between Preorders}
\citet{krner1977comparison} proved that $W \preceq_{d} V \Rightarrow W \preceq_{ln} V \Rightarrow W \preceq_{mc} V$ and gave examples to show that the reverse implications do not hold in general.
As \citet{americo2021channel} note, the degradation ($\preceq_d$), supermodularity ($\preceq_s$), and degradation/supermodularity ($\preceq_{ds}$) preorders are \emph{structural preorders}, in the sense that they only depend on the conditional probabilities that are defined by each channel. On the other hand, the \textit{less noisy} and \textit{more capable} preorders are concerned with information measures resulting from different distributions. It is trivial to see (directly from the definition) that the degradation preorder implies the degradation/supermodular preorder. \citet{americo2021channel} showed that the degradation/supermodular preorder implies the more capable preorder. The set of implications we have seen is schematically depicted in Figure \ref{fig:implications}.

\begin{figure}[H] \label{porelation}
	\centering
	\includegraphics[scale=0.25]{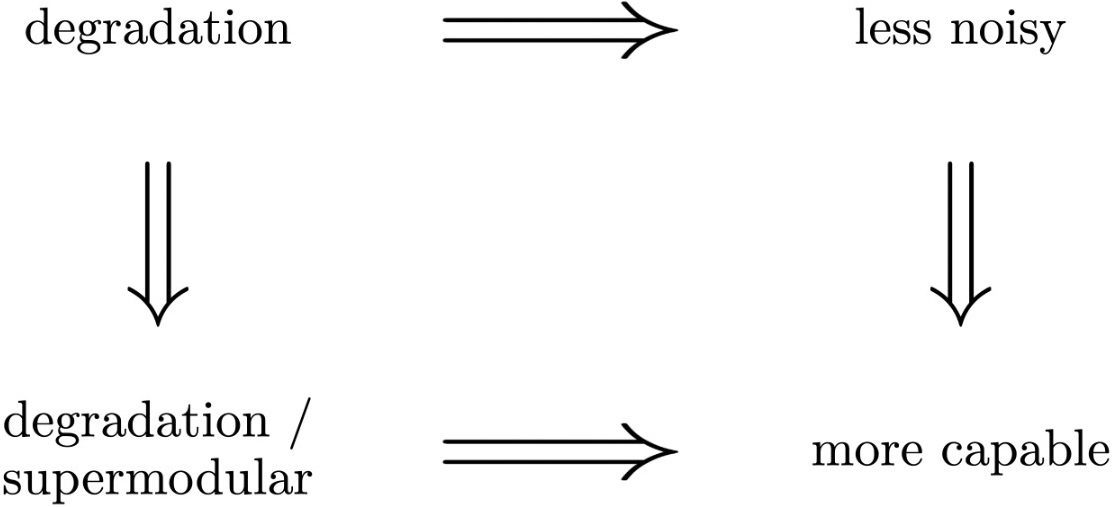}
	\caption{Implications satisfied by the preorders. The reverse implications do not hold in general.}\label{fig:implications}
\end{figure}

These relations between the preorders, for any set of variables $Y_1, ..., Y_n, T$, imply via the corresponding definitions that
\begin{align} \label{desigualdades1}
I_\cap^{d}(Y_1, ..., Y_n) \leq I_\cap^{ln}(Y_1, ..., Y_n) \leq I_\cap^{mc}(Y_1, ..., Y_n)
\end{align}
and
\begin{align} \label{desigualdades2}
I_\cap^{d}(Y_1, ..., Y_n) \leq I_\cap^{ds}(Y_1, ..., Y_n) \leq I_\cap^{mc}(Y_1, ..., Y_n).
\end{align}

These implications, in turn, imply the following result.
\begin{theorem}
The preorders $\preceq_{ln}$, $\preceq_{mc}$, and $\preceq_{ds}$, satisfy Kolchinsky's axioms.
\end{theorem}
\begin{proof}
Let $i \in \{1, ..., n\}$. Since any of the introduced preorders implies the \emph{more capable} preorder, it follows that they all satisfy the axiom of monotonicity of mutual information. Axiom 2 is trivially true since reflexivity is guaranteed by the definition of preorder. As for axiom 3, the rows of a channel corresponding to a variable $C$ taking a constant value must all be the same (and yield zero mutual information with any target variable $T$). From this, it is clear that any $Y_i$ satisfies $C \preceq Y_i$ for any of the introduced preorders, by definition of each preorder. To see that $Y_i \preceq Y=(Y_1, ..., Y_n)$ for the \emph{less noisy} and the \emph{more capable}, recall that for any $U$ such that $U \rightarrow T \rightarrow Y_i$ and $U \rightarrow T \rightarrow Y$, it is trivial that $I(U;Y_i) \leq I(U;Y)$, hence $Y_i \preceq_{ln} Y$. A similar argument may be used to show that $Y_i \preceq_{mc} Y$, since $I(T;Y_i) \leq I(T;Y)$. Finally, to see that $Y_i \preceq_{ds} (Y_1, ..., Y_n)$, note that $Y_i \preceq_d (Y_1, ..., Y_n)$ \citep{kolchinsky2022novel}, hence $Y_i \preceq_{ds} (Y_1, ..., Y_n)$. 
\end{proof}

\section{Optimization Problems}
We now focus on some observations about the optimization problems of the introduced II measures. All problems seek to maximize $I(Q;T)$ (under different constraints) as a function of the conditional distribution $p(q|t)$, equivalently with respect to the channel from $T$ to $Q$, which we will denote as $K^Q:=K^{Q|T}$. For fixed $p(t)$ -- as is the case in PID -- $I(Q;T)$ is a convex function of $K^Q$ \cite[Theorem 2.7.4]{cover1999elements}. For the degradation preorder, the supremum is achieved \citep{kolchinsky2022novel}; for the other preorders, this is not known, since the number of outcomes of $Q$ is not bounded.

As noted by \citet{kolchinsky2022novel}, the computation of \eqref{d} can be rewritten as an optimization problem using auxiliary variables such that it involves only linear constraints, and since the objective function is convex, its maximum is attained at one of the vertices of the admissible region. The computation of the other measures is not as simple, as shown in the following subsections. 

\subsection{The ``Less Noisy" Preorder}
To solve \eqref{ln}, we may use one of the necessary and sufficient conditions presented by \citet[Theorem 1]{makur2017less}. For instance, let $V$ and $W$ be two channels with input $T$, and $\Delta^{T-1}$ be the probability simplex of the target $T$. Then, $V \preceq_{ln} W$ if and only if, for any pair of distributions $p(t), q(t) \in \Delta^{T-1}$, the inequality 
\begin{equation}
\chi^2 (p(t)W || q(t)W) \geq \chi^2 (p(t)V || q(t)V) \label{eq:makur}
\end{equation}
holds, where $\chi^2$ denotes the $\chi^2$-distance\footnote{The $\chi^2$ distance between two vectors $u$ and $v$ of dimension $n$ is given by $\chi^2(u ||  v) = \sum_{i=1}^n (u_i - v_i)^2 / v_i$.} between two vectors. Notice that $p(t)W$ is the distribution of the output of channel $W$ for input distribution $p(t)$; thus, intuitively, the condition in \eqref{eq:makur} means that the two output distributions of the \textit{less noisy} channel are more different from each other than those of the other channel. Hence, computing $I_\cap^{ln}(Y_1, ..., Y_n)$ can be formulated as solving the problem
\begin{flalign*}
\hspace*{0.1\linewidth}\begin{aligned}
\sup_{K^Q} \quad & I(Q;T)\\
\textrm{s.t.} &\; K^Q \; \mbox{is a stochastic matrix,}  \\
  & \; \forall p(t), q(t) \in \Delta^{T-1} , \forall i \in \{1, ..., n\}, \; \chi^2 (p(t) K^{(i)} || q(t)K^{(i)}) \geq \chi^2 (p(t)K^Q  || q(t)K^Q ). \\
\end{aligned} &&
\end{flalign*}

Although the restriction set is convex since the $\chi^2$-divergence is an $f$-divergence, with $f$ convex \citep{csiszar2011information}, the problem is intractable because we have an infinite (uncountable) number of restrictions. One may construct a set $\mathcal{S}$ by taking an arbitrary number of samples $S$ of $p(t) \in \Delta^{T-1}$ to define the problem
\begin{flalign} \label{ln_approx}
\hspace*{0.1\linewidth}\begin{aligned}
\sup_{K^Q} \quad & I(Q;T), \\
\textrm{s.t.} &\; K^Q \; \mbox{is a stochastic matrix,}  \\
  &\; \forall p(t), q(t) \in \mathcal{S}, \forall i \in \{1, ..., n\}, \; \chi^2 (p(t)K^{(i)} || q(t)K^{(i)}) \geq \chi^2 (p(t)K^Q || q(t)K^Q). \\
\end{aligned} &&
\end{flalign}
The above problem yields an upper bound on $I_\cap^{ln}(Y_1, ..., Y_n)$. 

\subsection{The ``More Capable" Preorder}
To compute $I_\cap^{mc}(Y_1, ..., Y_n)$, we define the problem
\begin{flalign} \label{mc2}
\hspace*{0.1\linewidth}\begin{aligned}
\sup_{K^Q} \quad & I(Q;T)\\
\textrm{s.t.} &\; K^Q \; \mbox{is a stochastic matrix,}  \\
  &\; \forall p(t) \in \Delta^{T-1}, \forall i \in \{1, ..., n\}, \; I(Y_i;T) \geq I(Q;T), \\
\end{aligned} &&
\end{flalign}
which also leads to a convex restriction set, because $I(Q;T)$ is a convex function of $K^Q$. We discretize the problem in the same manner to obtain a tractable version
\begin{flalign} \label{mc_approx}
\hspace*{0.1\linewidth}\begin{aligned}
\sup_{K^Q} \quad & I(Q;T)\\
\textrm{s.t.} &\; K^Q \; \mbox{is a stochastic matrix,}  \\
  & \; \forall p(t) \in \mathcal{S}, \forall i \in \{1, ..., n\}, \; I(Y_i;T) \geq I(Q;T), \\
\end{aligned} &&
\end{flalign}
which also yields an upper bound on $I_\cap^{mc}(Y_1, ..., Y_n)$. 

\subsection{The ``Degradation/Supermodularity" Preorder}
The final introduced measure, $I_\cap^{ds}(Y_1, ..., Y_n)$, is given by
\begin{flalign}
\hspace*{0.1\linewidth}\begin{aligned}
\sup_{K^Q} \quad & I(Q;T)\\
\textrm{s.t.} &\; K^Q \; \mbox{is a stochastic matrix,}  \\
  & \; \forall i, K^Q \preceq_{ds} K^{(i)}. \\
\end{aligned} &&
\end{flalign}
To the best of our knowledge, there is currently no known condition to check if $K^Q~\preceq_{ds}~K^{(i)}$.

\section{Relation with Existing PID Measures}
\citet{griffith2014intersection} introduced a measure of II as
\begin{align} \label{griffithII}
I_\cap^{\triangleleft} (Y_1, ..., Y_n) := \max_Q I(Q;T), \text{ such that } \forall i \text{ } Q \triangleleft Y_i,
\end{align}
with the preorder relation $\triangleleft$ defined by $A \triangleleft B$ if $A=f(B)$ for some deterministic function $f$. That is, $I_\cap^{\triangleleft}$ quantifies redundancy as the presence of deterministic relations between input and target. If $Q$ is a solution of \eqref{griffithII}, then there exist functions $f_1, ..., f_n$, such that $Q = f_i(Y_i), i=1, ..., n$, which implies that, for all $i, T \rightarrow Y_i \rightarrow Q$ is a Markov chain. Therefore, $Q$ is an admissible point of the optimization problem that defines $I_\cap^{d} (Y_1, ..., Y_n)$, thus we have that $I_\cap^{\triangleleft} (Y_1, ..., Y_n) \leq I_\cap^{d} (Y_1, ..., Y_n)$.

\citet{barrett2015exploration} introduced the so-called \textit{minimum mutual information} (MMI) measure of bivariate redundancy as
$$ I_\cap^{\text{MMI}} (Y_1, Y_2) := \min \{I(T; Y_1), I(T; Y_2)\}.
$$ 
It turns out that, if $(Y_1, Y_2)$ is jointly Gaussian and $T$ is univariate, then most of the introduced PIDs in the literature are equivalent to this measure \citep{barrett2015exploration}. As noted by \citet{kolchinsky2022novel}, it may be generalized to more than two sources,
$$ I_\cap^{\text{MMI}} (Y_1, ..., Y_n) := \sup_Q I(Q;T), \text{ such that } \forall i \text{ } I(Q;T) \leq I(Y_i; T),$$
which allows us to trivially conclude that, for any set of variables $Y_1, ..., Y_n, T$,
$$I_\cap^{\triangleleft} (Y_1, ..., Y_n) \leq I_\cap^{d}(Y_1, ..., Y_n) \leq I_\cap^{mc}(Y_1, ..., Y_n) \leq I_\cap^{\text{MMI}}(Y_1, ..., Y_n).$$

One of the appeals of measures of II, as defined by \citet{kolchinsky2022novel}, is that it is the underlying preorder that determines what is intersection - or redundant - information. For example, consider the $I_\cap^\text{GH}$ measure \citep{griffith2015quantifying} in the $n=2$ case. Its solution, $Q$, satisfies $T \perp Q \, |\,  Y_1$ and $T \perp Q \, |\,  Y_2$, that is, if either $Y_1$ or $Y_2$ are known, $Q$ has no additional information about $T$. Such is not necessarily the case for the \textit{less noisy} or the \textit{more capable} II measures, that is, the solution $Q$ may have additional information about $T$ even when a source is known. However, $I_\cap^{d}$ and $I_\cap^{ln}$ satisfy the following property: any admissible point $Q$ of the optimization problem satisfies
$$ 
\forall i \in \{1, ..., n \}, \forall t \in S_T, I(Y_i; T=t) \geq I(Q; T=t),
$$
where $S_T$ is the support of $T$ and $I(T=t; Y_i)$ refers to the so-called specific information \citep{williams2010nonnegative, deweese1999measure}. That is, independently of the outcome of $T$, $Q$ has less specific information about $T=t$ than any source variable $Y_i$. This follows from the data processing inequality for $I_\cap^{d}$ and from the definition of the less noisy preorder for $I_\cap^{ln}$; it does not hold for $I_\cap^{mc}$ and $I_\cap^{ds}$ (e.g., $Q = Y_4$ in the example of Table~\ref{tab4} has more specific information than $Y_3$ about the third outcome of $T$). Such is not the case, for example, for $I_\cap^{\text{MMI}}$, which is arguably one of the reasons why it has been criticized for depending only on the amount of information, and not on its content \cite{kolchinsky2022novel}.
As mentioned, there is not much consensus as to what properties a measure of II should satisfy. The three proposed measures for partial information decomposition do not satisfy the so-called \textit{Blackwell property} \citep{bertschinger2014quantifying,rauh2017extractable}:
\begin{definition}
An intersection information measure $I_\cap(Y_1, Y_2)$ is said to satisfy the Blackwell property if the equivalence $Y_1 \preceq_d Y_2 \Leftrightarrow I_\cap(Y_1, Y_2) = I(T;Y_1)$ holds.
\end{definition}
This definition is equivalent to demanding that $Y_1 \preceq_d Y_2$, if and only if $Y_1$ has no unique information about $T$. Although the $(\Rightarrow)$ implication holds for the three proposed measures, the reverse implication does not, as shown by specific examples presented by \citet{krner1977comparison}, which we will mention below. If one defines the ``more capable property" by replacing the \textit{degradation} preorder with the \textit{more capable} preorder in the original definition of the Blackwell property, then measure $k$ satisfies the $(\Rightarrow)$ implication of the $k$ property, with $k$ referring to any of the three introduced intersection information measures; the reverse implication does not hold in general for $k = mc$ (for a uniform binary $T$, if $Y_1$ and $Y_2$ are the outputs of a binary erasure channel with erasure probability $h(0.1)$, where $h$ is the binary entropy function, and of a binary symmetric channel with crossover probability $0.1$, then $I_\cap^{mc}(Y_1, Y_2) = I(T;Y_1)$ \citep{polyanskiy2024information}, but $Y_1 \not\preceq_{mc} Y_2$). 

Also often studied in PID is the \textit{identity property} (IP) \citep{harder2013bivariate}. Let the target $T$ be a copy of the source variables, that is, $T = (Y_1, Y_2)$. An II measure $I_\cap$ is said to satisfy the IP if
$$I_\cap(Y_1, Y_2) = I(Y_1;Y_2).$$
Criticism was levied against this proposal for being too restrictive \citep{rauh2014reconsidering, james2018unique}. A less strict property was introduced by \cite{ince2017measuring}, under the name \textit{independent identity property} (IIP). If the target $T$ is a copy of the input, an II measure is said to satisfy the IIP if
$$I(Y_1;Y_2)=0 \;\; \Rightarrow \;\; I_\cap(Y_1, Y_2) = 0.$$
Note that the IIP is implied by the IP, but the reverse does not hold. It turns out that all the introduced measures, just like the degradation II measure, satisfy the IIP, but not the IP, as we will show. This can be seen from \eqref{desigualdades1}, \eqref{desigualdades2}, and the fact that $I_\cap^{mc}(Y_1, Y_2 \rightarrow (Y_1, Y_2))$ equals 0 if $I(Y_1;Y_2)=0$, as we argue now. Consider the distribution where $T$ is a copy of $(Y_1, Y_2)$, presented in Table \ref{copy_dist}.
\begin{table}[H]
\caption{Copy distribution}\label{copy_dist}
\vspace{0.2cm}
\begin{center}
    \begin{tabular}{|c|c|c|c|}
    \hline
    T & $Y_1$ & $Y_2$ & $p(t,y_1,y_2)$ \\
    \hline
    (0,0) & 0 & 0 & $p(T=(0,0))$ \\
    \hline
    (0,1) & 0 & 1 & $p(T=(0,1))$ \\
    \hline
    (1,0) & 1 & 0 & $p(T=(1,0))$ \\
    \hline
    (1,1) & 1 & 1 & $p(T=(1,1))$ \\
    \hline
\end{tabular}\quad
\end{center}
\end{table}
We assume that each of the 4 events has non-zero probability. In this case, channels $K^{(1)}$ and $K^{(2)}$ are given by
\[
  K^{(1)} = \left[\begin{array}{cc}
    1 & 0 \\
    1 & 0 \\
    0 & 1 \\
    0 & 1
  \end{array}\right], \quad
  K^{(2)} = \left[\begin{array}{cc}
    1 & 0 \\
    0 & 1 \\
    1 & 0 \\
    0 & 1
  \end{array}\right].
\]
Note that for any distribution $p(t)=\left[p(0,0), p(0,1), p(1,0), p(1,1)\right]$, if $p(1,0)=p(1,1)=0$, then $I(T;Y_1)=0$, which implies that, for any of such distributions, the solution $Q$ of \eqref{mc2} must satisfy $I(Q;T)=0$, thus the first and second rows of $K^Q$ must be the same. The same goes for any distribution $p(t)$ with $p(0,0)=p(0,1)=0$. On the other hand, if $p(0,0)=p(1,0)=0$ or $p(1,1)=p(0,1)=0$, then $I(T;Y_2)=0$, implying that $I(Q;T)=0$ for such distributions. Hence, $K^Q$ must be an arbitrary channel (that is, a channel that satisfies $Q \perp T$), yielding $I_\cap^{mc}(Y_1, Y_2) = 0$.  

Now recall the Gács-Korner \textit{common information} \citep{gacs1973common} defined as
\begin{equation}
\begin{aligned}
C(Y_1 \land Y_2) := &\sup_Q &H(Q) \label{eq:gk} \\
&\textrm{ s.t.} &Q \triangleleft Y_1 \\
& &Q \triangleleft Y_2
\end{aligned}
\end{equation}
We will use a similar argument and slightly change the notation to show the following result.
\begin{theorem}
Let $T = (X, Y)$ be a copy of the source variables. Then $I_\cap^{ln}(X,Y) = I_\cap^{ds}(X,Y) = I_\cap^{mc}(X,Y) = C(X \land Y)$.
\end{theorem}

\begin{proof}
As shown by \citet{kolchinsky2022novel}, $I_\cap^{d}(X,Y) = C(X \land Y)$. Thus, \eqref{desigualdades1} implies that $I_\cap^{mc}(X,Y) \geq C(X \land Y)$. The proof will be complete by showing that $I_\cap^{mc}(X,Y) \leq C(X \land Y)$. Construct the bipartite graph with vertex set $\mathcal{X} \cup \mathcal{Y}$ and edges $(x,y)$ if $p(x,y)>0$. Consider the set of \textit{maximally connected components} $MCC = \{CC_1, ..., CC_l\}$, for some $l \geq 1$, where each $CC_i$ refers to a maximal set of connected edges. Let $CC_i, i \leq l$, be an arbitrary set in $MCC$. Suppose the edges $(x_1, y_1)$ and $(x_1, y_2)$, with $y_1 \neq y_2$ are in $CC_i$. This means that the channels $K^X:=K^{X|T}$ and $K^Y:=K^{Y|T}$ have rows corresponding to the outcomes $T=(x_1, y_1)$ and $T=(x_1, y_2)$ of the form
\[
K^X = \begin{bmatrix} 
      & &  & \vdots & & & \\
    0 & \cdots & 0 & 1 & 0 & \cdots & 0 \\
    0 & \cdots & 0 & 1 & 0 & \cdots & 0 \\
      & &  & \vdots & & & \\
    \end{bmatrix},
\quad
K^Y = \begin{bmatrix} 
      & & & \vdots & & & & \\
    0 & \cdots &  0 & 1 & 0 & 0 & \cdots & 0 \\
    0 & \cdots &  0 & 0 & 1 & 0 & \cdots & 0 \\
      & & & \vdots & & & & \\
    \end{bmatrix}.
\]
Choosing $p(t)=[0, ..., 0, a, 1-a, 0, ..., 0]$, that is, $p(T=(x_1,y_1))=a$ and $p(T=(x_1,y_2))=1-a$, we have that, $\forall a \in [0,1], I(X;T)=0$, which implies that the solution $Q$ must be such that, $\forall a \in [0,1], I(Q;T)=0$ (from the definition of the \emph{more capable} preorder), which in turn implies that the rows of $K^Q$ corresponding to these outcomes must be the same, so that they yield $I(Q;T)=0$ under this set of distributions. We may choose the values of those rows to be the same as those rows from $K^X$ - that is, a row that is composed of zeros except for one of the positions whenever $T=(x_1,y_1)$ or $T=(x_1,y_2)$. On the other hand, if the edges $(x_1, y_1)$ and $(x_2, y_1)$, with $x_1 \neq x_2$, are also in $CC_i$, the same argument leads to the conclusion that the rows of $K^Q$ corresponding to the outcomes $T=(x_1, y_1)$, $T=(x_1, y_2)$, and $T=(x_2, y_1)$ must be the same. Applying this argument to every edge in $CC_i$, we conclude that the rows of $K^Q$ corresponding to outcomes $(x,y) \in CC_i$ must all be the same. Using this argument for every set $CC_1, ..., CC_l$ implies that if two edges are in the same CC, the corresponding rows of $K^Q$ must be the same. These corresponding rows of $K^Q$ may vary between different CCs, but for the same CC, they must be the same.

Hence, for any admissible $Q$, $p(q|t)$ depends on $t$ only through the index $V$ of the component $CC_i$ that contains $(X, Y)$, so $T \to V \to Q$ is a Markov chain. Since $V$ is a function of both $X$ and $Y$, it is admissible in \eqref{eq:gk}, so $H(V) \leq C(X \land Y)$, and $I(Q;T) \leq I(V;T) \leq H(V) \leq C(X \land Y)$. Hence $I_\cap^{ln}(X,Y)=I_\cap^{ds}(X,Y)=I_\cap^{mc}(X,Y)=C(X \land Y)$ if $T$ is a copy of the input.
\end{proof}

\citet{bertschinger2014quantifying} suggested what later became known as the (*) assumption, which states that, in the bivariate source case, any sensible measure of unique information should only depend on $K^{(1)}, K^{(2)}$, and $p(t)$. It is not clear that this assumption should hold for every PID. It is trivial to see that all the introduced II measures satisfy the (*) assumption. 

We conclude with some applications of the proposed measures to famous (bivariate) PID problems, with results shown in Table~\ref{tab1}. Due to channel design in these problems, the computation of the proposed measures is fairly trivial. We assume the input variables are binary (taking values in $\{0,1\}$), independent, and equiprobable. 
\begin{table}[H] 
\caption{Application of the measures to famous PID problems.}\label{tab1}
\begin{tabularx}{\textwidth}{CCCCCCC}
\toprule
Target & $I_\cap^{\triangleleft}$ & $I_\cap^d$ & $I_\cap^{ln}$ & $I_\cap^{ds}$ & $I_\cap^{mc}$ & $I_\cap^{MMI}$ \\
\midrule
$T~=~Y_1~\text{AND}~Y_2$ & 0 & 0.311 & 0.311 & 0.311 & 0.311 & 0.311\\
$T~=~Y_1~+~Y_2$ & 0 & 0.5 & 0.5 & 0.5 & 0.5 & 0.5\\
$T= Y_1$ & 0 & 0 & 0 & 0 & 0 & 0\\
$T~=~(Y_1,Y_2)$ & 0 & 0 & 0 & 0 & 0 & 1\\
\bottomrule
\end{tabularx}
\end{table}
We note that in these fairly simple toy distributions, all the introduced measures yield the same value. This is not surprising when the distribution $p(t,y_1, y_2)$ yields $K^{(1)} = K^{(2)}$, which implies that $I(T;Y_1)=I(T;Y_2)=I_\cap^{k}(Y_1, Y_2)$, where $k$ refers to any of the introduced preorders, as is the case in the $T=Y_1 \text{ AND } Y_2$ and $T=Y_1+Y_2$ examples. Less trivial examples lead to different values over the introduced measures. We present distributions that show that our three introduced measures lead to novel information decompositions by comparing them to the following existing measures: $I_\cap^{\triangleleft}$ from \citet{griffith2014intersection}, $I_\cap^{\text{MMI}}$ from \citet{barrett2015exploration}, $I_\cap^{\text{WB}}$ from \citet{williams2010nonnegative}, $I_\cap^{\text{GH}}$ from \citet{griffith2015quantifying}, $I_\cap^{\text{Ince}}$ from \citet{ince2017measuring}, $I_\cap^{\text{FL}}$ from \citet{finn2017pointwise}, $I_\cap^{\text{BROJA}}$ from \citet{bertschinger2014quantifying}, $I_\cap^{\text{Harder}}$ from \citet{harder2013bivariate} and $I_\cap^{\text{dep}}$ from \cite{james2018unique}. We use the {\tt dit} package \citep{james2018dit} to compute them as well as the code provided in \citep{kolchinsky2022novel}. Consider counterexample 1 by \cite{krner1977comparison} with $p=0.25, \epsilon=0.2, \delta=0.1$, given by
\[
  K^{(1)} = \left[\begin{array}{cc}
    0.25 & 0.75 \\
    0.35 & 0.65
  \end{array}\right], \quad
  K^{(2)} = \left[\begin{array}{cc}
    0.675 & 0.325 \\
    0.745 & 0.255
  \end{array}\right].
\]
These channels satisfy $K^{(2)} \preceq_{ln} K^{(1)}$ and $K^{(2)} \npreceq_{d} K^{(1)}$ \citet{krner1977comparison}. This is an example that satisfies, for whatever distribution $p(t)$, $I_\cap^{ln}(Y_1, Y_2)=I(T;Y_2)$. It is noteworthy to see that even though there is no degradation preorder between the two channels, we still have that $I_\cap^d(Y_1, Y_2)>0$ because there is some non-trivial channel $K^Q$ that satisfies $K^Q \preceq_{d} K^{(1)}$ and $K^Q \preceq_{d} K^{(2)}$. We present various PID under different measures, after choosing $p(t)=[0.4, 0.6]$ (which yields $I(T;Y_2) \approx 0.004$) and assuming $p(t,y_1, y_2) = p(t)p(y_1|t)p(y_2|t)$.
\begin{table}[H] 
\caption{Different decompositions of $p(t,y_1,y_2)$.}\label{tab3}
\begin{tabularx}{\textwidth}{CCCCCCCCCCCCC}
\toprule
$I_\cap^{\triangleleft}$ & $I_\cap^d$ & $I_\cap^{ln}$ & $I_\cap^{ds}$ & $I_\cap^{mc}$ & $I_\cap^{\text{MMI}}$ & $I_\cap^{\text{WB}}$ & $I_\cap^{\text{GH}}$ & $I_\cap^{\text{Ince}}$ & $I_\cap^{\text{FL}}$ & $I_\cap^{\text{BROJA}}$ & $I_\cap^{\text{Harder}}$ & $I_\cap^{\text{dep}}$\\
\midrule
0 & 0.002 & 0.004 & * & 0.004 & 0.004 & 0.004 & 0.002 & 0.003 & 0.047 & 0.003 & 0.004 & 0\\
\bottomrule
\end{tabularx}
\end{table}
We write $I_\cap^{ds} =$ * because we don't yet have a way to find the `largest' $Q$, such that $Q \preceq_{ds} K^{(1)}$ and $Q \preceq_{ds} K^{(2)}$. See counterexample 2 by \cite{krner1977comparison} for an example of channels $K^{(1)}, K^{(2)}$ that satisfy $K^{(2)} \preceq_{mc} K^{(1)}$ but $K^{(2)} \npreceq_{ln} K^{(1)}$, leading to different values of the proposed II measures. An example of $K^{(3)}, K^{(4)}$ that satisfy $K^{(4)} \preceq_{ds} K^{(3)}$ but $K^{(4)} \npreceq_{d} K^{(3)}$ is presented by \citet[page 10]{americo2021channel}, given by
\[
  K^{(3)} = \left[\begin{array}{cc}
    1 & 0 \\
    0 & 1 \\
    0.5 & 0.5
  \end{array}\right], \quad
  K^{(4)} = \left[\begin{array}{cc}
    1 & 0 \\
    1 & 0 \\
    0.5 & 0.5 
  \end{array}\right].
\]
There is no stochastic matrix $K^U$, such that $K^{(4)} = K^{(3)} K^U$, but $K^{(4)} \preceq_{ds} K^{(3)}$ because $K^{(4)} = \diamond_{1,2} K^{(3)}$. Using \eqref{eq:makur} one may check that there is no \emph{less noisy} relation between the two channels \footnote{Compute \eqref{eq:makur} with $V=K^{(4)}, W = K^{(3)}$, $p(t)=[0,0,1]$ and $q(t)=[0.1, 0.1, 0.8]$ to conclude that $K^{(4)} \npreceq_{ln} K^{(3)}$, then switch the roles of $V$ and $W$ and set to $p(t) = [0,1,0]$ and $q(t) = [0.1, 0, 0.9]$ to conclude that $K^{(3)} \npreceq_{ln} K^{(4)}$.}. We present the decomposition of $p(t,y_3, y_4)=p(t)p(y_3|t)p(y_4|t)$ for the choice of $p(t) = [0.3, 0.3, 0.4]$ (which yields $I(T;Y_4) \approx 0.322$) in Table \ref{tab4}.
\begin{table}[H] 
\caption{Different decompositions of $p(t,y_3,y_4)$.}\label{tab4}
\begin{tabularx}{\textwidth}{CCCCCCCCCCCCC}
\toprule
$I_\cap^{\triangleleft}$ & $I_\cap^d$ & $I_\cap^{ln}$ & $I_\cap^{ds}$ & $I_\cap^{mc}$ & $I_\cap^{\text{MMI}}$ & $I_\cap^{\text{WB}}$ & $I_\cap^{\text{GH}}$ & $I_\cap^{\text{Ince}}$ & $I_\cap^{\text{FL}}$ & $I_\cap^{\text{BROJA}}$ & $I_\cap^{\text{Harder}}$ & $I_\cap^{\text{dep}}$\\
\midrule
0 & 0 & $0^*$ & 0.322 & 0.322 & 0.322 & 0.193 & 0 & 0 & 0.058 & 0 & 0 & 0\\
\bottomrule
\end{tabularx}
\end{table}
We write $I_\cap^{ln}=0^*$ because we conjecture, after some numerical experiments based on \eqref{eq:makur}, that the `largest' channel that is less noisy than both $K^{(3)}$ and $K^{(4)}$ is a channel that satisfies $I(Q;T)=0$ \footnote{We tested all $3 \times 3$ row-stochastic matrices whose entries take values in $\{0, 0.1, 0.2, ..., 0.9, 1\}$ with all distributions $p(t)$ and $q(t)$ whose entries take values in the same set.}.
\section{Conclusion and future work}
In this paper, we introduced three new measures of \textit{intersection information} for the \textit{partial information decomposition} (PID) framework, based on preorders between channels implied by the degradation/Blackwell preorder. The new measures were obtained from the preorders by following the approach recently proposed by  \citet{kolchinsky2022novel}. The main contributions and conclusions of the paper can be summarized as follows: 
\begin{itemize}
    \item We showed that a measure of intersection information generated by a preorder that satisfies the axioms by \citet{kolchinsky2022novel} also satisfies the Williams-Beer axioms  \citep{williams2010nonnegative}.
     \item As a corollary of the previous result, the proposed measures satisfy the Williams-Beer axioms and can be extended beyond two sources.
    \item We demonstrated that, if there is a degradation preordering  between the sources, the measures coincide in their decomposition.  Conversely, if there is no degradation preordering (only a weaker preordering) between the source variables, the proposed measures lead to novel finer information decompositions that capture different, finer information.
    \item We showed that the proposed measures do not satisfy the \textit{identity property} (IP) \citep{harder2013bivariate}, but satisfy the \textit{independent identity property} (IIP)  \citep{ince2017measuring}.
    \item We formulated the optimization problems that yield the proposed measures and derived bounds by relating them to existing measures.
\end{itemize}

Finally, we believe this paper opens several avenues for future research, thus we point at several directions to be pursued in upcoming work:

\begin{itemize}
    \item Investigating conditions to verify if two channels, $K^{(1)}$ and $K^{(2)}$, satisfy $K^{(1)} \preceq_{ds} K^{(2)}$. 
    \item \citet{kolchinsky2022novel} showed that when computing $I_\cap^{d}(Y_1, ..., Y_n)$, it is sufficient to consider variables $Q$ with support size, at most, $\sum_i |S_{Y_i}|-n+1$, as a consequence of the admissible region of $I_\cap^{d}(Y_1, ..., Y_n)$ being a polytope. Such is not the case with the \textit{less noisy} or the \textit{more capable} measures, hence it is not clear if it suffices to consider $Q$ with the same support size. This is a direction of future research.
    \item Studying under which conditions the different \textit{intersection information} measures are continuous. 
    \item Implementing the different introduced measures, by addressing the corresponding optimization problems.
    \item Considering the usual PID framework, but instead of decomposing $I(T;Y) = H(Y)-H(Y|T)$, where $H$ denotes Shannon's entropy, one can consider other mutual informations, induced by different entropy measures, such as the \textit{guessing entropy} \citep{massey1994guessing} or the \textit{Tsallis entropy} \citep{tsallis1988possible}. See the work of \citet{americo2021channel} for other core-concave entropies that may be decomposed under the introduced preorders, as these entropies are consistent with the introduced preorders. 
    \item Another line for future work is to define measures of union information with the introduced preorders, as suggested by \citet{kolchinsky2022novel}, and study their properties.
    \item As a more long-term research direction, it would be interesting to study how the approach studied in this paper can be extended to quantum information, where the fact that partial quantum information can be negative may open possibilities or create difficulties \citep{Horodecki2005}.
\end{itemize}

\textbf{Acknowledgments:} We thank Artemy Kolchinsky for helpful discussions and suggestions.

\textbf{Funding:} This research was partially funded by: FCT -- \textit{Fundação para a Ciência e a Tecnologia}, under grants number SFRH/BD/145472/2019 and UIDB/50008/2020; Instituto de Telecomunicações; Portuguese Recovery and Resilience Plan, through project C645008882-00000055 (NextGenAI, CenterforResponsibleAI).

\bibliographystyle{unsrtnat}
\bibliography{references}  






\end{document}